\documentclass[11pt]{article}
	
	\newcommand{\blind}{1}
	
    \makeatletter
    \renewcommand\section{\@startsection {section}{1}{\z@}%
                                       {-3.5ex \@plus -1ex \@minus -.2ex}%
                                       {2.3ex \@plus.2ex}%
                                       {\normalfont\fontfamily{phv}\fontsize{16}{19}\bfseries}}
    \renewcommand\subsection{\@startsection{subsection}{2}{\z@}%
                                         {-3.25ex\@plus -1ex \@minus -.2ex}%
                                         {1.5ex \@plus .2ex}%
                                         {\normalfont\fontfamily{phv}\fontsize{14}{17}\bfseries}}
    \renewcommand\subsubsection{\@startsection{subsubsection}{3}{\z@}%
                                        {-3.25ex\@plus -1ex \@minus -.2ex}%
                                         {1.5ex \@plus .2ex}%
                                         {\normalfont\normalsize\fontfamily{phv}\fontsize{14}{17}\selectfont}}
    \makeatother

	\usepackage{amsmath}
	\usepackage{graphicx}
	\usepackage{enumerate}
	\usepackage{xcolor}
	\usepackage{natbib} 
	\usepackage{url} 
    
    \usepackage{amsthm,amsfonts,amssymb}
    \usepackage{amsthm}
    \newtheorem{proposition}{Proposition}

    \usepackage{algorithm}
    \usepackage{algpseudocode}
    \algtext*{EndIf}
    \algtext*{EndFor}
    
    \usepackage{graphicx}
    \usepackage{array}
    \usepackage{tabularx}
    \usepackage{multirow}
    \usepackage{booktabs}
    \usepackage{makecell}
    \usepackage{subfig}
    \usepackage{adjustbox}
    
    \newcolumntype{Y}{>{\centering\arraybackslash}X}
    
    \usepackage[english]{babel}
    \usepackage[shortlabels]{enumitem}
    \usepackage[normalem]{ulem}
    \usepackage{textcomp}
    \usepackage{gensymb}
    \usepackage{verbatim}
    \usepackage{soul}
    
    \usepackage[table]{xcolor}
    \usepackage{pifont}
    
    \definecolor{YesColor}{RGB}{220,245,225}
    \definecolor{NoColor}{RGB}{245,225,225}
    
    \newcommand{\yes}{\cellcolor{YesColor}\ding{51}}
    \newcommand{\no}{\cellcolor{NoColor}}
    
    \usepackage{siunitx}
    \DeclareSIUnit{\USD}{\$}
    \DeclareSIUnit{\MWh}{MWh}

    \usepackage{titlesec}

    \titleformat{\paragraph}[block]
      {\normalfont\normalsize\bfseries}
      {\theparagraph}{1em}{}
    \titlespacing*{\paragraph}{0pt}{1ex}{0.5ex}
    
\begin{document}
		
		\def\spacingset#1{\renewcommand{\baselinestretch}%
			{#1}\small\normalsize} \spacingset{1}
		
		\if0\blind
		{
			\title{\bf From Market Signals to Maintenance Decisions: Grid-Informed Maintenance Optimization in Offshore Wind Farms}
			\author{Keyi Wang$^1$, Robert Mieth$^1$ and Ahmed Aziz Ezzat$^1$ \\
			$^1$Department of Industrial \& Systems Engineering, \\Rutgers, The State University of New Jersey, USA.}
			\date{}
			\maketitle
		} \fi
        
		\if1\blind
		{

            \title{\bf From Market Signals to Maintenance Decisions: Grid-Informed Maintenance Optimization in Offshore Wind Farms}
			\author{Author information is purposely removed for double-blind review}
			
\bigskip
			\bigskip
			\bigskip
			\begin{center}
				{\LARGE\bf From Market Signals to Maintenance Decisions: Grid-Informed Maintenance Optimization in Offshore Wind Farms\\}
            
            {\vspace{0.5cm} Keyi Wang$^1$, Robert Mieth$^1$ and Ahmed Aziz Ezzat$^{1,\dagger}$ \\
			$^1$Department of Industrial \& Systems Engineering, \\Rutgers, The State University of New Jersey, USA.\\
            $^\dagger$Corresponding author: aziz.ezzat@rutgers.edu}
			\end{center}
			\medskip
		} \fi
		\bigskip
		
	\begin{abstract}
The reliable and cost-effective integration of large-scale offshore wind farms into power systems is largely influenced by maintenance operations. To date, power producers primarily rely on \textit{farm-level} information (e.g., health signals, power generation, weather conditions) for maintenance planning of wind turbine fleets. We propose to consider additional \textit{grid-level} information 
in maintenance operations, such as generation curtailment and market pricing. 
Periods of high generation curtailment or negative electricity pricing 
can be 
viewed as opportunities for preventive maintenance, since assets are fully or partially shut down due to grid-imposed constraints. 
We first analytically establish that access to market information has non-negative economic value for maintenance scheduling under a stylized setting, providing a theoretical foundation for treating market and grid information as a maintenance opportunity. We then formulate a grid-informed opportunistic maintenance scheduling model as a receding-horizon mixed integer linear program that jointly leverages both farm- and grid-level signals under realistic conditions to determine a maintenance plan for offshore wind turbines.
Considering a 600-MW offshore wind farm in the United States, our proposed method reduces annual maintenance costs by up to 10\%, while enhancing asset availability and resource utilization, relative to advanced maintenance models that solely rely on farm-level information.
	\end{abstract}
			
	\noindent%

	\spacingset{1.5} 

\section{Introduction} \label{S:intro}

Global offshore wind energy capacity has grown at an annual average rate of $10\%$ over the past decade, reaching $83.2$ Gigawatts (\unit{GW}), and accounting for $7.3\%$ of total global wind capacity as of the end of 2024 \citep{GWEC}. 
Owing to its numerous economic and environmental benefits, offshore wind is viewed as a key ingredient of future power systems, generating large-scale, reliable, and carbon-free electricity. Yet, ensuring the economic viability of large-scale offshore wind projects requires advanced maintenance models that can ensure asset availability and reliability, especially as power grids must increasingly coordinate the integration of diverse energy resources and growing demand under operational, market, and network constraints.

Operations and maintenance (O\&M) planning is central to the economic viability of offshore wind farms 
\citep{EWEA}. 
Maintenance scheduling in offshore wind farms entails planning around a highly uncertain marine environment 
and restricted site accessibility \citep{HOST}. 
To ensure cost-effective maintenance under such challenging conditions, offshore O\&M scheduling increasingly relies on ``\textit{opportunistic maintenance},''
wherein 
multiple maintenance tasks 
may be ``grouped'' at times of opportunity (e.g., persistent periods of provisional site access, or low wind resources), such that their joint execution would reduce overall O\&M costs 
\citep{Besnard1}. 

The vast majority of existing opportunistic maintenance optimization methods rely on what we refer to as \textit{farm-informed maintenance triggers}. 
These approaches leverage one or more sources of information that are ``endogenous'' to the wind farm to guide the scheduling of maintenance tasks, including asset health 
\citep{Asset}, weather conditions \citep{Accessibility}, the availability of maintenance resources \citep{HOST}, and wind-dependent production potential \citep{Besnard2}.   
This paper addresses a blind spot in existing opportunistic maintenance optimization by tapping into an additional class of maintenance triggers that are responsive to the power grid that connects the wind farm, and hence, are ``exogenous'' to the wind farm.

\subsection{Grid-informed maintenance triggers}
Grid-level information remains underused for maintenance scheduling in offshore wind farms, as pointed out in recent literature 
\citep{McMorland,mieth2026strategic}. If modeled properly, grid-informed maintenance scheduling can unlock new maintenance opportunities that are driven by grid operations.
In this work, we refer to these farm-exogenous sources of information as \textit{grid-informed maintenance triggers}, and particularly study two of these triggers: generation curtailment and electricity pricing signals. 


Generation curtailment is becoming more prevalent as wind energy penetration increases. It 
can be broadly classified into system curtailment and market-based curtailment. 
System curtailment results from grid congestion that temporarily or systematically limits the injection of power generated by the wind farm. 
Market-based curtailment, in contrast, is an economic dispatch decision made by the operator, typically due to low demand or insufficient operational flexibility.
Recent data shows that curtailment events are increasing in frequency and magnitude across electricity system operators worldwide. 
In the U.S., the wholesale market in the Electric Reliability Council of Texas (ERCOT) grid curtailed $5\%$ and $9\%$ of its total available wind and solar generation in 2022. By 2035, wind and solar curtailments are projected to increase to $13\%$ and $19\%$, respectively \citep{UScurtail}.

From the perspective of an (offshore) wind power producer, curtailment and negative electricity prices typically reduce realized revenues. 
However, we postulate that the same event can be an opportunity for operational optimization. 
During curtailment events, turbines are intentionally under-dispatched or fully shut down by the system operator to reduce a corresponding proportion of power generation. 
These periods of forced downtime can be strategically leveraged by wind farm operators to perform preventive maintenance without sacrificing energy yield. 
Moreover, such preventive opportunistic maintenance actions offer grid benefits by improving the reliability of generation assets and, in turn, enhancing the accuracy of grid-side wind power production forecasts.

An important remaining question is whether the temporal patterns of curtailment or negative price periods align with the requirements of maintenance tasks in offshore wind farms. 
Our empirical analysis of the U.S. East Coast power system, explained in greater detail in  Section \ref{S:numerical_exp},
suggests that there are sufficiently long windows of curtailment or low electricity prices that make them suitable maintenance triggers. For example, we see that high curtailment levels (up to 30\% of total generation capacity) typically last from 4 to 12 hours in length, which can be ideal for minor to medium repairs. 
These types of repairs comprise 
the vast majority of maintenance tasks in offshore wind farms \citep{Dinwoodie, Carroll}. Longer curtailment intervals (12+ hours) occur about 4\% of the time, which is coincident with the less frequent, major repairs. Understandably, the frequency, length, and magnitude of such curtailment events will likely continue to increase as power systems integrate more renewable generation capacity \citep{UScurtail, NESO}. 

\subsection{Literature Review}
Opportunistic maintenance for offshore wind farms has been extensively studied in recent years. As summarized in Table~\ref{tab:literature}, weather conditions are widely recognized as a key maintenance trigger, as they affect both power generation \citep{Zhang, Liyang, POSYDON} and site accessibility \citep{Besnard2, dalgic2015advanced, HOST, petersen2026accessibility, Sahin2026AccessibilityAware}. 
Another frequently studied maintenance trigger is asset health. Here, degradation or reliability indicators are used to trigger condition-based maintenance (CBM) \citep{wang2013prognosis,Shafiee, dui2017cost, Zhou, yousefi2019optimization, rozas2024condition}. 
Several studies further integrate weather and asset health information, recognizing that turbine degradation and failure risk evolve dynamically under varying environmental conditions \citep{Accessibility,Erguido1, Li1, POSYDON}. 
Maintenance and logistical resource information can also be viewed as opportunity triggers, including the availability 
of vessels, equipment, and spare parts \citep{Besnard1, Besnard2, Erguido1, HOST, li2024operation}.
\begin{table}[t]
\caption{Representative list of opportunistic O\&M models for wind turbine fleets.}  
\renewcommand{\arraystretch}{1}
\setlength{\tabcolsep}{3pt}
\centering
\footnotesize

\begin{tabularx}{\textwidth}{|l|Y|Y|Y||Y|Y|}
\hline
\multirow{3}{*}{Research Effort}
& \multicolumn{3}{c||}{
    \makecell{Farm-informed\\[-2pt]Maintenance Triggers}}
& \multicolumn{2}{c|}{
    \makecell{Grid-informed\\[-2pt]Maintenance Triggers}} \\
\cline{2-6}
& \makecell{Weather\\[-2pt]conditions}
& \makecell{Asset\\[-2pt]health}
& \makecell{Maintenance\\[-2pt]resources}
& \makecell{Electricity\\[-2pt]prices}
& \makecell{Generation\\[-2pt]curtailment} \\
\hline

\citet{Besnard1}
& \yes & \no & \yes & \no & \no \\

\citet{Besnard2}
& \yes & \no & \yes & \no & \no \\

\citet{Shafiee}
& \yes & \yes & \no & \no & \no \\

\citet{Yildirim}
& \yes & \yes & \yes & \yes & \no \\

\citet{Erguido1}
& \yes & \yes & \yes & \no & \no \\

\citet{Zhou}
& \yes & \yes & \no & \no & \no \\

\citet{Zhang}
& \yes & \no & \yes & \no & \no \\

\citet{Liyang}
& \yes & \yes & \yes & \no & \no \\

\citet{Li1}
& \yes & \yes & \no & \no & \no \\

\citet{HOST}
& \yes & \no & \yes & \yes & \no \\

\citet{Fallahi}
& \yes & \yes & \yes & \yes & \no \\

\citet{STOCHOS}
& \yes & \yes & \yes & \yes & \no \\

\citet{POSYDON}
& \yes & \yes & \yes & \yes & \no \\

\citet{Orla}
& \yes & \yes & \yes & \no & \yes \\

\hline
\textbf{GIMS (Proposed Method)}
& \yes & \yes & \yes & \yes & \yes \\
\hline

\end{tabularx}
\label{tab:literature}
\end{table}

Methods that treat grid information (such as electricity prices and curtailment) as maintenance triggers are extremely sparse.
\citet{Yildirim} jointly optimize turbine maintenance scheduling and dispatch using production and electricity pricing information. 
\citet{zhang2021two} consider the joint scheduling of asset maintenance and power production, integrating pre-determined power supply commitments from the day-ahead electricity market into maintenance planning. \citet{HOST} integrate electricity price information in their maintenance models. Their approach does not explicitly treat pricing information as a maintenance trigger, but rather as an additional source of input uncertainty for the simulation of wind farm revenue. Furthermore, none of the above references 
consider generation curtailment as a maintenance trigger. 

Perhaps the earliest mention of considering market signals in opportunistic maintenance settings is by \citet{McMorland}. 
Their study proposes a conceptual framework for leveraging market signals such as curtailment and pricing, but does not develop a detailed model for optimizing the use of such information. The closest effort to the one presented in this paper is the recent work by \citet{Orla} who explicitly leverage curtailment information in an O\&M scheduling framework. Their method entails 
an \textit{if-then-else} cost-benefit analysis where the cost to perform a repair action is compared with curtailment compensation offered by the grid operator. 
Maintenance is then assigned when this compensation exceeds the repair cost. 
While this approach 
has a number of limitations. 
First, it adopts a greedy decision strategy whereby maintenance is triggered immediately by a curtailment event, potentially missing favorable opportunities in the future. Second, maintenance opportunities can be formed by multiple interacting factors, including site accessibility and production levels. 
A direct comparison of 
repair cost with compensation payments overlooks these interactions. 
Hence a maintenance action may satisfy the cost-benefit analysis, but may result in unfavorable economics if it is scheduled during periods of high production 
or conditions with a high risk of inaccessibility and maintenance delays. 

A well-established approach to formulate opportunistic maintenance scheduling models is 
through a mixed integer linear program (MILP) \citep{Besnard1, Besnard2, Shafiee, Yildirim, Erguido1, Zhou, Zhang, HOST, STOCHOS, POSYDON}, wherein maintenance decisions are represented by binary variables, complemented by auxiliary variables for power generation and allocation of maintenance resources. 
Our paper, too, follows this modeling approach.
Alternative modeling approaches include simulation-based methods \citep{byon2011simulation, Orla}, dynamic programming \citep{Accessibility,soltani2024structured}, and more recently, decision-focused machine learning methods \citep{kazemian2024attention}. 

\subsection{Contributions}
We postulate that 
grid information 
are a largely untapped opportunity for improving O\&M scheduling in (offshore) wind farms. 
To this end, 
we propose a mathematical model for maintenance optimization that formally leverages these grid-informed maintenance triggers. 
We summarize our main contributions as follows:
\begin{itemize}
    \item We  analytically establish that access to market information, such as curtailment and pricing signals, can have non-negative economic value for maintenance
scheduling under stylized conditions, providing a theoretical foundation for treating market and system information as a maintenance opportunity.
    \item We develop GIMS (short for the \underline{G}rid \underline{I}nformed \underline{M}aintenance \underline{S}cheduling), a multi-horizon scheduling model for opportunistic maintenance in offshore wind farms. 
    GIMS is formulated 
    as a mixed integer linear program that jointly considers endogenous, farm-level information (asset health, weather, and power generation) and exogenous grid-level information (generation curtailment and electricity pricing) and their 
    complex impacts on economic and operational performance of offshore wind farms. 

    \item  We evaluate GIMS using a realistic test case from the Northeastern U.S., a geographical region with significant offshore wind potential and development plans. Our experiments show that considering curtailment and pricing signals as grid-level opportunity triggers reduces maintenance costs by up to $10\%$ for a $600$-\si{\mega\watt} wind farm, while substantially enhancing asset availability and resource utilization, relative to methods that do not fully leverage such exogenous information.  
    
\end{itemize}
 
The remainder of this paper is organized as follows. 
The theoretical analysis, mathematical formulation, and solution algorithm of GIMS are presented in Section~\ref{S:Opt_model}.
The data and experimental setup are introduced in Section~\ref{S:numerical_exp}. In Section~\ref{S:res_ana}, we present and discuss our results.
Section~\ref{S:concl}
concludes the paper and outlines future research directions.

\section{Grid-Informed Maintenance for Offshore Wind Farms} \label{S:Opt_model}
We take the perspective of a profit-maximizing offshore wind farm operator that seeks revenue from selling power to the grid and incurs cost from turbine maintenance. 
Figure~\ref{fig:GIMS flowchart} illustrates the conceptual workflow of our method. 
At the farm level, asset health directly informs maintenance scheduling decisions through projected turbine remaining useful life (RUL). Maintenance scheduling constitutes the central decision variable in the framework, linking power generation and maintenance resource dispatch. 
Once maintenance actions are scheduled, maintenance crews and vessels are dispatched accordingly, resulting in O\&M costs and power output reduction due to turbine downtime. Weather conditions simultaneously impact power generation and turbine accessibility.
Maintenance resources, such as crew, equipment, and vessel availability, impose constraints on crew dispatch.
\begin{figure}[!htbp]
    \centering
    \includegraphics[width=.5\linewidth]{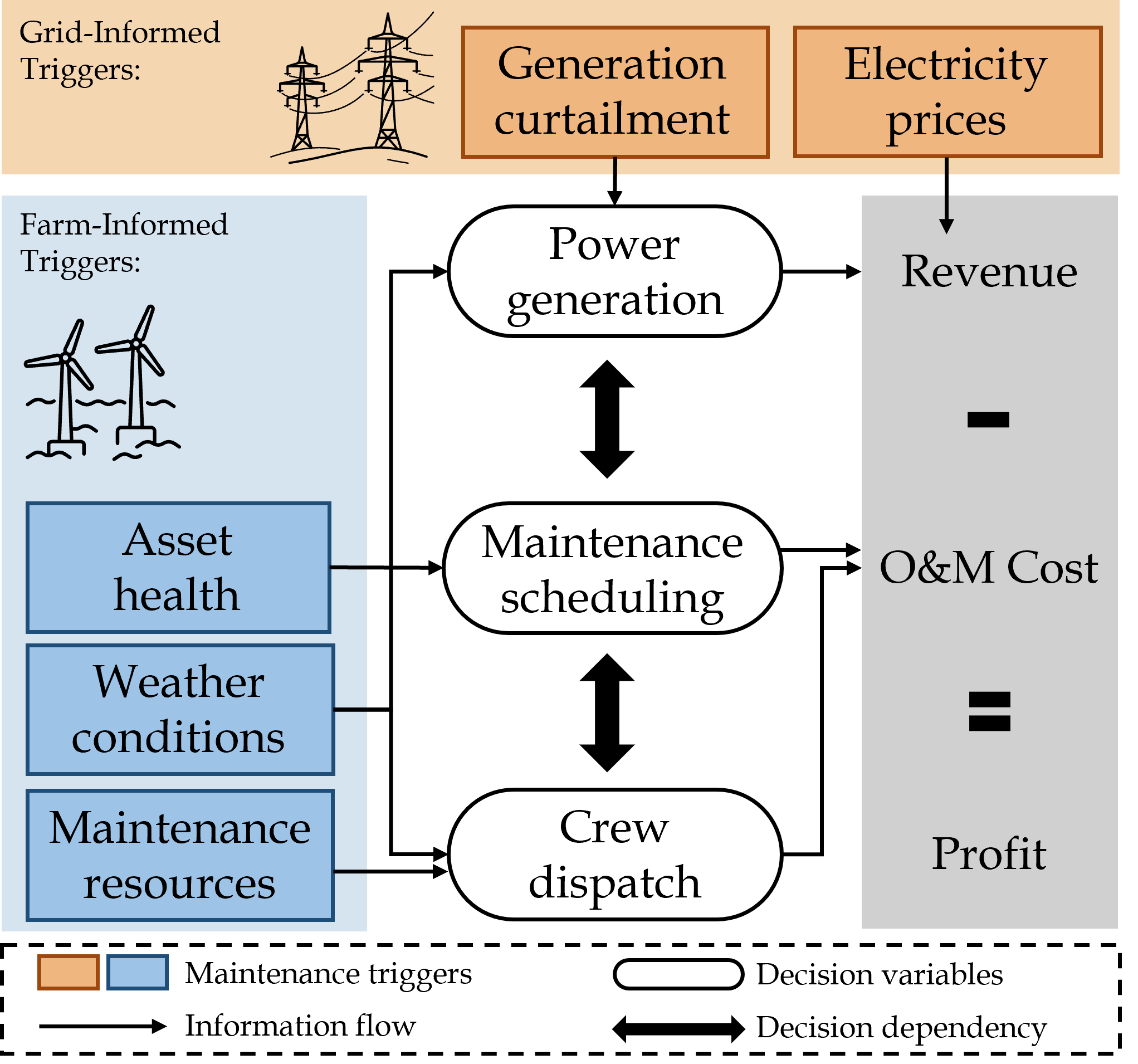}
    \caption{Illustration of 
    the GIMS framework. Grid-informed triggers 
    affect electricity generation and revenues. Farm-informed triggers 
    jointly govern generation, maintenance scheduling, and crew dispatch. The objective is to maximize the wind power producer's profit by balancing revenue and O\&M costs.}
    \label{fig:GIMS flowchart}
\end{figure}

In parallel, grid-informed triggers enter the decision process as exogenous operational and economic factors. Generation curtailment directly limits allowable power output, potentially forcing turbines to operate below rated capacity or remain offline. Electricity prices determine profit from electricity generation and define the trade-off between power production and allocating turbine downtime for maintenance. 
By jointly modeling these interactions, GIMS captures the coupling between grid conditions and O\&M decisions.

In the following subsections, we first 
provide a theoretical foundation for treating grid and market signals as maintenance opportunities. We then develop an optimization framework for realizing this value under realistic environmental and operational constraints.

\subsection{
Economic Value of Grid Information 
under a Stylized Setting}
\label{S:analytical_property}
A wind turbine requires a maintenance task to be scheduled on day $d\in\mathcal D=\{1,\ldots,N_{\mathcal D}\}$. Let $C_d\in[0,1]$ denote the curtailment ratio on day $d$, which defines the fraction of potential wind power generation that is curtailed by the grid operator, and let $K_d$ 
denote the corresponding electricity price on day $d$. For the theoretical analysis in this section, we impose assumptions $\mathcal{A}1$-$\mathcal{A}3$ to construct a stylized setting that isolates the economic value of incorporating grid and market information into maintenance planning. We subsequently relax these assumptions in Section \ref{s:formulation} where we develop the full optimization model under realistic and more general maintenance and operational conditions.
\begin{enumerate}[label=({$\mathcal{A}$}{{\arabic*}})]
\item Daily curtailment ratios and electricity prices $\{(C_d,K_d)\}_{d\in\mathcal D}$ are assumed to be independent and identically distributed realizations of a random bivariate vector $(C,K)$, with $C \in [0,1]$ and $K \geq 0$. 
\item A single wind turbine is considered for maintenance. 
\item Factors affecting maintenance scheduling other than curtailment and electricity prices (e.g., weather and maintenance resources) are not explicitly modeled. 
Instead, their effects are summarized through two fixed quantities: the nominal daily power generation of the turbine in the absence of maintenance and curtailment, $P^{\text{nom}} > 0$, and the generation loss associated with the maintenance action, $0 < \Delta G \le P^{\text{nom}}$. 
\end{enumerate}

In the absence of maintenance, 
the curtailment reduced daily generation \(P_d\) and the corresponding daily revenue \(V_d\) are given by Equation (\ref{eq:revenue_baseline}).
\begin{equation}
P_d=(1-C_d)P^{\mathrm{nom}} \quad , 
\quad
V_d
=
K_d(1-C_d)P^{\mathrm{nom}}.
\label{eq:revenue_baseline}
\end{equation}

We consider two maintenance scheduling strategies. The \textit{grid-blind strategy} selects maintenance day \(d^{\mathrm B}\) independently of the grid information ($C_d, K_d$). Under $\mathcal{A}1-\mathcal{A}3$, all days are \textit{ex ante} statistically equivalent to this strategy; hence, selecting any day yields the same expected maintenance-induced revenue loss. In contrast, the \textit{grid-informed strategy} observes ($C_d, K_d$) and selects the day with the lowest daily revenue, 
i.e., $d^{\mathrm I}=\arg\min_{d\in\mathcal D}V_d$.

When maintenance is performed, 
the maximum generation 
is reduced 
to $P^{\mathrm{nom}}-\Delta G$. 
Under the stylized model, realized generation under maintenance is therefore the minimum of the post-curtailment generation and the maintenance-constrained generation, that is, $\min\left\{(1-C_d)P^{\mathrm{nom}},P^{\mathrm{nom}}-\Delta G\right\}$, yielding the maintenance-induced revenue loss on day \(d\):
\begin{align}
    \Delta V_d
    &=
    K_d
    \left[
    P_d
    -
    \min\left\{
    (1-C_d)P^{\mathrm{nom}},
    P^{\mathrm{nom}}-\Delta G
    \right\}
    \right]
    \nonumber\\
    &=
    K_d
    \left(
    \Delta G-C_dP^{\mathrm{nom}}
    \right)^+,
    \label{eq:analytical_loss}
\end{align}
where $(x)^+=\max\{x,0\}$. 
Conceptually, Equation (\ref{eq:analytical_loss}) captures the interaction between maintenance- and curtailment-induced losses. 
When $C_dP^{\mathrm{nom}}<\Delta G$, existing curtailment absorbs only part of the generation loss associated with maintenance, resulting in an incremental loss of $\Delta G-C_dP^{\mathrm{nom}} > 0$. Conversely, when $C_dP^{\mathrm{nom}}\geq \Delta G$, existing curtailment fully absorbs the maintenance-induced generation loss, and $\Delta V_d = 0$.

Let $\Delta V^{\mathrm B}=\Delta V_{d^{\mathrm B}}$ and $\Delta V^{\mathrm I}=\Delta V_{d^{\mathrm I}}$
denote the maintenance-induced revenue losses under the grid-blind and grid-informed strategies, respectively. Proposition \ref{prop:aware_vs_blind} establishes a sufficient condition under which the grid-informed strategy weakly reduces the expected maintenance-induced revenue loss. 
\begin{proposition}[Economic value of grid information in maintenance scheduling]
\label{prop:aware_vs_blind}
Let \(f_{V\mid C}(v\mid c)\) denote the conditional density of the wind power producer's revenue given curtailment $C = c$. Suppose that, for any \(c_2>c_1\), the likelihood ratio
$\frac{f_{V\mid C}(v\mid c_2)}
     {f_{V\mid C}(v\mid c_1)}
\text{ is nonincreasing in } v$,
over the support of \(f_{V\mid C}(v\mid c_1)\). Then, the expected maintenance-induced revenue loss under the grid-informed strategy is no greater than that under the grid-blind strategy:
\begin{equation}
\boxed{
\mathbb E[\Delta V^{\mathrm I}]
\le
\mathbb E[\Delta V^{\mathrm B}].
}
\label{eq:aware_vs_blind}
\end{equation}
\end{proposition}
A proof is provided in Appendix~\ref{app:proof_aware_blind}. Proposition~\ref{prop:aware_vs_blind} establishes a sufficient condition under which incorporating grid information into maintenance scheduling does not increase the expected maintenance-induced revenue loss. 
The condition imposed in Proposition~\ref{prop:aware_vs_blind} is a monotone likelihood ratio characterizing the probabilistic relationship between curtailment and revenue: For \(c_2>c_1\), a non-increasing likelihood ratio implies that higher-curtailment states are relatively more likely to be associated with lower revenue. This relationship aligns with the operational reality of electric systems where curtailment arises from transmission congestion or local oversupply and therefore coincides with lower electricity prices.

Section \ref{s:formulation} builds on this theoretical insight by relaxing assumptions $\mathcal{A}1$–$\mathcal{A}3$ and incorporating realistic operational constraints and maintenance conditions, such as those shown in Figure \ref{fig:GIMS flowchart}, into the proposed maintenance model.

\subsection{Optimization Model} \label{s:formulation}
We consider a grid-connected offshore wind farm comprising a set of turbine assets indexed by $i\in\mathcal{I}$. Maintenance decisions are made over two planning horizons: a short-term horizon (STH) with day-ahead outlook and hourly resolution $t\in\mathcal{T}=\{1,...,24\}$, and a long-term planning horizon (LTH), defined in daily periods $d\in\mathcal{D}=\{1,...,N_{\mathcal{D}}\}$. The LTH starts immediately after the short-term horizon and extends for up to $N_{\mathcal{D}}$ days. The planning problem is solved under a daily rolling-horizon procedure provided in Section \ref{S:algorithm}. GIMS makes two primary sets of decisions: power generation and asset maintenance:
\begin{itemize}[leftmargin=*]
    \item\textit{Power generation:} $p_{t,i},\  p_{d,i}\in \mathbb{R}_{\ge 0}$ denote the average power generated by the $i$th turbine during hour $t$ and day $d$, respectively. 
    \item\textit{Maintenance decisions:} $m_{t,i},\ m_{d,i} \in \{0,1\}$ indicate if maintenance on the $i$th turbine is scheduled to start at hour~$t$ ($m_{t,i}=1$) and on day $d$ ($m_{d,i}=1$), respectively. 
\end{itemize}

\subsubsection{Objective Function}
The objective function, expressed in \eqref{eq:obj_fun}, is the total of the day-ahead profit $z$ and the future profits realized over the LTH, such that $z_d$ denotes the profit of day $d$.
\begin{equation}\label{eq:obj_fun}
    \max_{m_{t,i}, m_{d,i}} z + \sum_{d\in \mathcal{D}}z_d.
\end{equation}

The STH profit $z$ is expressed in (\ref{STH_prof}) and is calculated as operating revenue less maintenance-related costs. Operating revenue is obtained from the product of the day-ahead electricity  market price $\kappa_{t}$ and generated power $p_{t,i}$. The maintenance-related costs include five components: \textit{(i)} crew costs incurred at an hourly crew rate $C^c$, where $x_{t,i}\in\{0,1\}$ indicates if a crew is dispatched to the $i$th turbine at hour $t$, \textit{(ii)} repair cost (explained below), \textit{(iii)} maintenance extension cost (explained below), \textit{(iv)} overtime expenses incurred at an hourly overtime rate of $C^o$, where $o$ denotes the total overtime hours accrued during the short-term horizon, and \textit{(v)} vessel rental costs charged, on a daily basis, at a rate of $C^v$, where $v\in\{0,1\}$ denotes if a vessel is chartered for the STH ($v=1$). 
\begin{multline} \label{STH_prof}
z=
\sum_{i \in \mathcal{I}}
\Bigg[
\sum_{t \in \mathcal{T}}
\bigg(
\overbrace{\kappa_{t} \cdot p_{t,i}}^{\text{operating revenue}}
-
\overbrace{C^c \cdot x_{t,i}}^{\text{crew cost}}
-
\overbrace{(1-\rho_i)\cdot[s_{i}\cdot C^{\text{PM}}+(1-s_{i})\cdot C^{\text{CM}}]\cdot m_{t,i}}^{\text{repair cost}}
\bigg)
\\
-
\overbrace{\Big(
C^{v}\cdot \pi_{i}
+\kappa_{1} \cdot \tilde{p}_{1,i}
\Big)}^{\text{maintenance extension cost}}
\Bigg]
-
\overbrace{C^o \cdot o}^{\text{overtime cost}}
-
\overbrace{C^v \cdot v}^{\text{vessel cost}}.
\end{multline}

The repair cost in \eqref{STH_prof} depends on whether the maintenance action is preventive or corrective. If turbine $i$ remains operational at the onset of maintenance, the action is considered a \textit{preventive maintenance} (PM) and incurs a cost $C^{\text{PM}}$.   Otherwise, if the turbine has already failed, the action is a \textit{corrective maintenance} (CM) and incurs a cost
$C^{\text{CM}}$. The turbine status variable $s_i \in \{0,1\}$, as defined in 
\eqref{eq:state}, indicates whether the turbine's remaining useful life (RUL) $\lambda_i$ has been fully expended, signaling a failed turbine. 
\begin{equation}
\label{eq:state}
s_i = \mathbb{I}\{\lambda_i > 1\},
\end{equation}
If a maintenance task was initiated on the previous day and remains ongoing, as indicated by $\rho_i = 1$, its repair cost is not incurred again; consequently, the repair-cost term in \eqref{STH_prof} reduces to zero. The parameter $\rho_i$ is not a decision variable, but a binary indicator computed \textit{post ante} from the previous day's maintenance decisions.

The maintenance extension cost in \eqref{STH_prof} captures the expenses associated with maintenance tasks that cannot be completed within the STH. These interruptions are represented by the auxiliary variable $\pi_i \in \{0,1\}$ which indicates whether maintenance of turbine $i$ extends beyond the STH. As shown in \eqref{STH_prof}, 
the maintenance extension cost comprises two components: the additional vessel rental cost $C^{v}\cdot \pi_{i}$ needed to complete the repair task, and the expected revenue loss during the extension period $\kappa_{1} \cdot \tilde{p}_{{d=1},i}$, such that $\kappa_1$ is the electricity price of the first day in the LTH, and $\tilde{p}_{{d=1},i}$ is the corresponding lost production during that day. $\tilde{p}_{1,i}$ is calculated as $\tilde{p}_{1,i} = a_{d=1,i} f_{d=1,i}\mbox{R}(2T\pi_i + r_i)$, where $a_{d=1,i} f_{d=1,i}\mbox{R}$ is the production on the first day of the LTH (see \ref{S: maintenance-power generation}). $T$ and $r_i$ denote the one-way travel time and the remaining maintenance duration, respectively, both in hours. Note that $T$ is a constant determined by logistical route information, while $r_i$ is an auxiliary variable that will be computed using maintenance decisions in \eqref{eq:incomplete2s}-\eqref{eq:incompletes}.


The profit over the LTH, $z_d$ is expressed in (\ref{LTH_prof}). Its structure is similar to the STH profit in \eqref{STH_prof} with daily-resolution parameters and variables, 
but the maintenance extension cost is replaced by the deferred maintenance term. This term discourages pre-mature maintenance decisions by assigning greater value to decisions that extend productive asset life. 
The parameter $C^\lambda$ controls the magnitude of this term and represents the marginal economic value assigned to preserving one additional day of usable turbine life. 
\begin{multline} \label{LTH_prof}
z_d = 
\sum_{i \in \mathcal{I}} \bigg[ \overbrace{\kappa_{d}\cdot p_{d,i}}^{\text{operating revenue}}
-
\overbrace{C^c\cdot \theta_{d,i} \cdot m_{d,i}}^{\text{crew cost}}
-
\overbrace{(1-\rho_i)\cdot[s_{d,i}\cdot C^{\text{PM}}+(1-s_{d,i})\cdot C^{\text{CM}}]\cdot m_{d,i}}^{\text{repair cost}}
\\
- 
{\overbrace{C^{\lambda}\cdot (\lambda_i - d)\cdot m_{d,i}}^{\text{deferred maintenance}}}\Big]
-
\overbrace{C^o \cdot o_{d}}^{\text{overtime cost}}
-
\overbrace{C^v \cdot v_{d}}^{\text{vessel cost}}
\quad {\forall ~ d \in \mathcal{D}}.
\end{multline}

\subsubsection{Farm-Grid Interaction Constraints}
\label{S:farm-grid interaction}
The central feature of GIMS is the interaction between turbine-level power availability and grid-imposed curtailment. Turbine maintenance affects the amount of power that can be generated by the wind farm, while curtailment limits the amount of power that can be delivered to the grid. We therefore first characterize turbine-level generation and then impose the corresponding farm-level grid constraints.

\paragraph{Asset Maintenance-Power Generation Interaction} 
\label{S: maintenance-power generation}
The potential power generation of each turbine depends on its realized wind conditions and availability. We define the weather-dependent capacity factor of each turbine $i$ in hour $t$ as $f_{t,i} \in [0,1]$. Let \(a_{t,i}\in\{0,1\}\) indicate whether turbine \(i\) is available during hour \(t\), where \(a_{t,i}=1\) denotes availability. This availability variable will be linked directly to the maintenance decisions through \eqref{eq:availabilitys}-\eqref{eq:availability2s}. Then, 
the hourly power generation $p_{t,i}$ is expressed as in (\ref{eq:yaw_power_STH}), where $\mbox{R}$ is the rated capacity of a single turbine (in MW). 
When the turbine is unavailable ($a_{t,i}=0$) due to failure or maintenance, it cannot generate power. 
\begin{equation}\label{eq:yaw_power_STH}
    p_{t,i}\leq a_{t,i} f_{t,i} \mbox{R} \quad \forall t \in \mathcal{T}, i \in \mathcal{I}.
\end{equation}

Similarly, we denote $f_{d,i} \in [0,1]$ as the average capacity factor 
of turbine $i$ on day $d$, and let \(a_{d,i}\in\{0,1\}\) denote its availability during that day. Accordingly, the daily power production $p_{d,i}$ is bounded by:
\begin{equation}\label{eq:yaw_power_LTH}
    p_{d,i}\leq 24 a_{d,i} f_{d,i}\mbox{R}   \quad \forall d \in \mathcal{D}, i \in \mathcal{I}.
\end{equation}

For preventive maintenance scheduled during the LTH, the turbine 
becomes unavailable for an expected maintenance duration \(\theta_{d,i}\). Accordingly, the corresponding reduction in available operating time affects the average power generation as follows:  
\begin{equation} \label{eq:mainttsL}
p_{d,i}\leq  f_{d,i} (24 -\theta_{d,i}  m_{d,i}) \mbox{R} \quad \forall d \in \mathcal{D}, i \in \mathcal{I}.
\end{equation}

\paragraph{Grid Curtailment-Power Generation Interactions} 
We model the curtailment requirements issued by the grid operator as follows.
The nominal (pre-curtailed and pre-maintenance) power generation of the offshore wind farm at time $t$ is the aggregate of individual wind turbine outputs given as:
\begin{equation} 
\label{eq:nom_power_STH}
P^{\mathrm{nom}}_t = \mbox{R}  \sum_{i \in \mathcal{I}} f_{t,i}
\quad \forall t \in \mathcal{T}.
\end{equation}

If the wind farm is curtailed, the admissible power generation of the wind farm is restricted by a curtailment level $C_t \in [0,1]$, denoting the fraction of nominal power generation curtailed at time $t$. Accordingly, the ``post-curtailed'' generation is defined as:
\begin{equation}\label{eq:cur_power_STH}
P_t
= (1 - C_t) P^{\mathrm{nom}}_t
\quad \forall t \in \mathcal{T}.
\end{equation}
Analogously, for the LTH, the nominal daily power generation and the corresponding grid-imposed generation limit are given by:
\begin{equation}\label{eq:nom_power_LTH}
P^{\mathrm{nom}}_d = 24 \mbox{R}  \sum_{i \in \mathcal{I}} f_{d,i}
\quad\forall d \in \mathcal{D},
\end{equation}
\begin{equation}\label{eq:cur_power_LTH}
P_d
= (1 - C_d)  P^{\mathrm{nom}}_d
\quad \forall d \in \mathcal{D},
\end{equation}
where \(C_d\in[0,1]\) denotes the daily curtailment ratio. 

As such, the aggregate power actually generated by the wind farm is constrained by the grid-imposed generation limit:
\begin{align}
\sum_{i \in \mathcal{I}} p_{t,i}
&\leq P_t,
\quad \forall t \in \mathcal{T}, \label{eq:power_STH} \\
\sum_{i \in \mathcal{I}} p_{d,i}
&\leq P_d
\quad \forall d \in \mathcal{D}.\label{eq:power_LTH}
\end{align}

Equations \eqref{eq:yaw_power_STH}-\eqref{eq:yaw_power_LTH} and \eqref{eq:power_STH}-\eqref{eq:power_LTH} therefore establish two distinct upper bounds on wind-farm generation. 
The farm-informed bound is directly influenced by turbine availability and maintenance decisions, whereas the grid-informed bound is imposed exogenously by the grid operator through the curtailment level. Consequently, the effective generation limit is governed by the more restrictive of the two. 

This interaction is central to GIMS. When the grid-imposed limit is below the generation that would otherwise be available from the farm, a reduction in turbine availability due to maintenance may be absorbed, in whole or in part, by existing curtailment. In such cases, scheduling maintenance during periods of high curtailment can reduce the incremental production and revenue loss associated with taking a turbine offline. This mechanism is the operational counterpart of the theoretical result established in Section \ref{S:analytical_property}. 


\subsubsection{Farm-Informed Maintenance Constraints}
The remaining constraints characterize the turbine health and maintenance resources that determine the farm-informed generation bound established in Section \ref{S:farm-grid interaction}. These constraints 
are common in the condition-based wind farm maintenance literature (see, e.g., \cite{STOCHOS, POSYDON}). They collectively determine 
when and whether maintenance can be performed, and how maintenance activities affect turbine availability.
\paragraph{Turbine health constraints} 
GIMS prioritizes turbines with a high failure risk for maintenance.
We introduce the maintenance requirement variable $\chi_i \in\{0,1\}$ which indicates whether turbine $i$ should be considered for maintenance if its RUL, $\lambda_i$, falls below the RUL threshold $N^\lambda$. 
\begin{equation}\label{requirment}
\chi_i = \mathbb{I}\{\lambda_i \le N^\lambda\}.
\end{equation}

The equality constraint in \eqref{eq:maints} links the maintenance requirement variable $\chi_i$ to the maintenance scheduling decisions $m_{t,i}$, $m_{d,i}$. 
If $\chi_i = 1$, a maintenance task must be scheduled for wind turbine $i$, in either  the STH or the LTH. 
\begin{equation} \label{eq:maints}
\sum_{t \in \mathcal{T}} m_{t,i} + \sum_{d \in \mathcal{D}} m_{d,i} = \chi_i \quad \forall i \in \mathcal{I}.
\end{equation}

\paragraph{Maintenance and Logistical Resource Constraints} \label{subsec:log} 
A turbine for which a maintenance task is scheduled to start at time $t$ in the STH remains shut down (as indicated by $u_{\tilde{t},i}=1$)  
for a duration that is equal to the minimum of the remaining time within the STH, which is $24 - t$, and the expected maintenance duration~$\theta_{t,i}$:  
\begin{equation} \label{eq:downs}
\sum_{\tilde{t} = t}^{\min\{24,\, t + \theta_{t,i}\}} u_{\tilde{t},i}
\ge
\min\{24 - t,\, \theta_{t,i}\} \cdot m_{t,i}
\quad \forall t \in \mathcal{T},\, i \in \mathcal{I}.
\end{equation}

If this maintenance is not finished within the STH, the remaining maintenance time $r_i$ is defined in \eqref{eq:incomplete2s}. 
The binary variable $\pi_i$ in \eqref{eq:incompletes} is set to $1$ whenever $r_i > 0$. We enforce this through the big-$\mbox{M}$ constraint in \eqref{eq:incompletes}, where $\mbox{M}_1$ is a large positive number. The complementary counterpart is not required because $\pi_i$ carries a positive penalty in the objective function; thus, $\pi_i=0$ is optimal whenever $r_i=0$. 
\begin{align} 
r_i &\ge \sum_{t \in \mathcal{T}} m_{t,i} \cdot [\theta_{t,i} - 24 + t]^+
\quad \forall i \in \mathcal{I} \label{eq:incomplete2s} \\
r_i &\le \mbox{M}_1 \cdot \pi_i
\quad \forall i \in \mathcal{I}.
 \label{eq:incompletes}
\end{align}

In both the STH and LTH, vessels are rented on a daily basis if at least one maintenance task is scheduled on that day, as enforced by constraints \eqref{eq:vessel1s} and \eqref{eq:vessel2s}, respectively, where $\mbox{M}_2$ is another large positive number. 
\begin{align} 
    \mbox{M}_2\cdot v &\geq \sum_{t \in \mathcal{T}} \sum_{i \in \mathcal{I}} m_{t,i}, \label{eq:vessel1s} \\
  \mbox{M}_2 \cdot v_{d} &\geq \sum_{i \in \mathcal{I}}m_{d,i} \quad \forall d \in \mathcal{D}.\label{eq:vessel2s}
\end{align}

The total number of maintenance crews available at any time is bounded by $N^c$: 
\begin{equation} \label{eq:crew2s}
\sum_{i \in \mathcal{I}} x_{t,i} \leq N^c \quad \forall t \in \mathcal{T}.
\end{equation}

The total working hours of the maintenance crews in the STH are limited by $N^w$ hours per crew per day. 
When this limit is exceeded, overtime hours are incurred and tracked by variable $o$.
Overtime hours are compensated at a higher rate $C^o$. The total number of daily overtime hours is capped by $N^o$ 
as enforced by constraints \eqref{eq:workhourss} and \eqref{eq:overtime1s}:
\begin{align} 
\sum_{t \in \mathcal{T}} \sum_{i \in \mathcal{I}} x_{t,i} &\leq N^c N^w + o, \label{eq:workhourss} \\
o &\leq N^c N^o.  \label{eq:overtime1s}
\end{align}

A similar set of crew working-hour and overtime constraints is proposed for the LTH. Unfinished maintenance tasks from the STH are given priority at the onset of the LTH ($d=1$). 
The term $\theta_{1,i} m_{1,i} + (2 T \pi_{i}+r_{i})$ in \eqref{eq:workhours2s} represents the total required work hours, including both newly scheduled maintenance and the remaining work from unfinished tasks.
\begin{equation} \label{eq:workhours2s}
\sum_{i \in \mathcal{I}} \left[\theta_{1,i} m_{1,i} + (2  T  \pi_{i}+r_{i})\right] \leq N^c N^w + o_{1},
\end{equation}
For the remainder of the LTH (up to $N_{\mathcal{D}}$), the total crew working hours are constrained as:
\begin{equation} \label{eq:workhours3s}
\sum_{i \in \mathcal{I}} \left[\theta_{d,i} m_{d,i}\right] \leq N^c  N^w + o_{d} \quad \forall d \in \{2, ..., N_{\mathcal{D}}\}.
\end{equation}
Finally, the total number of daily overtime hours is bounded by $N^o$:
\begin{equation} \label{eq:overtime2s}
o_{d} \leq N^c N^o \quad \forall d \in \mathcal{D}.
\end{equation}

\paragraph{Turbine availability constraints}
As established in Section \ref{S:farm-grid interaction}, the maintenance decisions affect power generation through turbine availability. 
For the STH, turbine availability is constrained by:
\begin{equation} \label{eq:availabilitys}
a_{t,i}
\le
s_i (1 - \rho_i)
+
\sum_{\tilde{t} \in \mathcal{T}}
\frac{24 - \tilde{t}}{24 - t}\, m_{\tilde{t},i}
\quad \forall t \in \mathcal{T}, i \in \mathcal{I},
\end{equation}
The first term reflects whether a turbine has failed ($s_i = 0$) or is already under maintenance from a previous day ($\rho_i = 1$). 
Once the RUL has been fully expended, the turbine remains unavailable and can only be restored to availability after a maintenance task is assigned, as enforced by the second term.
The analogous LTH constraint is: 
\begin{equation} \label{eq:availability2s}
a_{d,i}
\le
s_{d,i} (1 - \rho_i)
+
\sum_{\tilde{d} \in \mathcal{D}}
\frac{N_{\mathcal{D}} - \tilde{d}}{N_{\mathcal{D}} - d}\, m_{\tilde{d},i},
\quad \forall d \in \mathcal{D},\ i \in \mathcal{I}.
\end{equation}

Additionally, a turbine under maintenance ($u_{t,i} = 1$) remains unavailable and occupies a crew, as enforced by \eqref{eq:maintts} and \eqref{eq:crews}.
\begin{equation} \label{eq:maintts}
a_{t,i} \le 1 - u_{t,i}
\quad \forall t \in \mathcal{T},\, i \in \mathcal{I}.
\end{equation}
\begin{equation} \label{eq:crews}
x_{t,i} \ge u_{t,i}
\quad \forall t \in \mathcal{T},\, i \in \mathcal{I}.
\end{equation}

Finally, a maintenance task requires crew travel to and from the wind farm in addition to the maintenance work itself.
When maintenance is scheduled for turbine \(i\), either during the STH or as a continuation from a previous day, the total crew dispatch time must account for both travel legs and at least one hour of maintenance: 
\begin{equation} \label{eq:travel time}
\sum_{t \in \mathcal{T}} x_{t,i}
\ge
(2 T + 1)  \left( \sum_{t \in \mathcal{T}} m_{t,i} + \rho_i \right)
\quad \forall i \in \mathcal{I}.
\end{equation}

\subsection{Solution Algorithm} \label{S:algorithm}
GIMS is implemented in an iterative rolling-horizon algorithm to generate dynamic maintenance schedules over the entire planning horizon, as summarized in Algorithm~\ref{alg:GIMS}. The notation for the model variables and parameters appearing in Algorithm~\ref{alg:GIMS} is defined in Section \ref{s:formulation}, except for the wind speed and wave height, $\nu$ and $\eta$, respectively, their corresponding safety limits 
$\nu_{max}$ and $\eta_{max}$, and the turbine-specific repair time, $\tau_i$. 


At each iteration \(\ell\), the algorithm considers a short-term horizon (STH) in hourly resolution \(\mathcal{T}\) and a long-term horizon (LTH) \(\mathcal{D}\) spanning the subsequent \(N_D\) days. Forecasts for wind speed \(\nu_t, \nu_d\) and wave height \(\eta_t, \eta_d\) are provided for the STH and LTH. 
Wind speed forecasts are converted to turbine-specific capacity factors \(f_{t,i}\) and \(f_{d,i}\), which are then used to estimate generation 
following 
Eqs. \eqref{eq:yaw_power_STH}-\eqref{eq:mainttsL} and \eqref{eq:nom_power_STH}-\eqref{eq:nom_power_LTH} in Section \ref{s:formulation}. 
Market signal forecasts, namely electricity prices \(\kappa_t, \kappa_d\) and curtailment levels \(C_t, C_d\), are used to update the generation limits, $P_t$ and $P_d$ through Eqs. \eqref{eq:cur_power_STH}-\eqref{eq:cur_power_LTH}. These limits are shaped by both the farm- and grid-informed generation bounds, as discussed in Section \ref{S:farm-grid interaction}. 

Turbine-specific health information, \(\lambda_i\), 
determines the turbine's maintenance requirement \(\chi_i\) (
Recall Eq. \eqref{requirment}), as well as its operational status \(s_i\) and \(s_{d,i}\) (
Eq. \eqref{eq:state}). These quantities inform the maintenance and availability constraints in Eqs.  
\eqref{eq:maints} and \eqref{eq:availabilitys}–\eqref{eq:availability2s}, respectively. 
The expected repair time \(\tau_i\), together with the wind and wave forecasts $\nu$ and $\eta$, and their corresponding safety limits $\nu_{max}$ and $\eta_{max}$, determine the weather- and accessibility-adjusted maintenance durations \(\theta_{t,i}\) and \(\theta_{d,i}\). These are pre-computed using the approach of \cite{STOCHOS},  
and subsequently incorporated into the maintenance and logistical resource constraints in Section \ref{subsec:log}. 


The resulting MILP is solved using Gurobi by maximizing the combined STH and LTH profits, as expressed in \eqref{eq:obj_fun}, subject to 
a relative optimality gap tolerance of 0.01\% and a computational time limit of 1,800 seconds are imposed. The resulting solution provides the optimal hourly maintenance schedule for the STH and daily maintenance plan for the LTH, where $\{m^\ell_{t,i}\}_{t\in\mathcal{T},i\in\mathcal{I}}$ denotes the day-ahead maintenance schedule at iteration $\ell$. 

Importantly, only the STH maintenance decisions are executed at each iteration. 
After the STH decisions are implemented, the algorithm evaluates the resulting maintenance status of each turbine. If a scheduled maintenance task is completed within the STH, the corresponding turbine state is reset and updated with its new health information. If the task is not completed within the STH, the maintenance is extended into the subsequent rolling horizon, and the associated repair information is updated to account for the unfinished task. The algorithm then advances the rolling horizon by one day, incorporates the newly available forecasts and updated turbine states, and resolves the MILP. This procedure is repeated until the end of the planning horizon. The hourly STH schedules generated across all iterations are concatenated to obtain the executed maintenance schedule $\mathcal{S}^*$
, which is used to evaluate the performance of GIMS relative to the benchmark approaches.



\begin{algorithm}[!htbp]
\caption{An Iterative Solution Algorithm for GIMS}
\label{alg:GIMS}
\begin{algorithmic}[1]

\State \textbf{Input} set ${\mathcal{I}}$, $N_{\mathcal{D}}$, optimization horizon $L$, initial turbine parameters $\{\lambda_i, \tau_i\}_{i\in\mathcal{I}}$
\Statex\hspace {2.4em} cost parameters: $C^{\lambda}, C^{\text{PM}}, C^{\text{CM}}, C^c, C^o, C^v$
\Statex \hspace {2.4em} operational parameters: $R, T, N^\lambda, N^c, N^w, N^o, \nu_{max}, \eta_{max}$

\For{$\ell=1,\dots,L-N_D$}
    \State Set
        $\mathcal{T}=\{24(\ell-1)+1,\dots,24\ell\}$,
        $\mathcal{D}=\{\ell+1,\dots,\ell+N_D\}$

    \State Input wind and wave information 
    $\{\nu_{t}, \eta_{t}\}_{t\in\mathcal{T}}$ and $\{\nu_{d}, \eta_{d}\}_{d\in\mathcal{D}}$

    \State Input pricing and curtailment information: $\{\kappa_{t}, C_{t}\}_{t\in\mathcal{T}}$ and $\{\kappa_{d}, C_{d}\}_{d\in\mathcal{D}}$

    \State Update farm status:
        $\{\chi_i\}_{i\in\mathcal{I}}$,
        $\{f_{t,i},s_i,\theta_{t,i},P_t\}_{t\in\mathcal{T},i\in\mathcal{I}}$,
        $\{f_{d,i},s_{d,i},\theta_{d,i},P_d\}_{d\in\mathcal{D},i\in\mathcal{I}}$

    \State Solve the MILP. Store the optimal solution
        $S^\ell_{\mathrm{STH}}
        =\{m^\ell_{t,i}\}_{t\in\mathcal{T},i\in\mathcal{I}}$

    \For{$i\in\mathcal{I}$ such that
        $\sum_{t\in\mathcal{T}}m^\ell_{t,i}>0$}
        \If{$\pi_i=0$} \quad Set $\rho_i=0$. Input new RUL $\lambda_i$ \Comment{task finished, reset turbine}
        \Else \quad Set $\rho_i=1$. $\tau_i\leftarrow 2T+r_i$ \Comment{maintenance extension}
        \EndIf
    \EndFor
\EndFor

\State \Return maintenance schedules $S^*=\{S^1_{\text{STH}},\dots,S^{L-N_D}_{\text{STH}}\}$

\end{algorithmic}
\end{algorithm}


\section{Data and Experimental Setup} \label{S:numerical_exp}
This section describes the data and experimental setup for evaluating GIMS. 

\subsection{Data Description} \label{S:data}
Our datasets include weather conditions, power system data, turbine health information, and wind power generation models, all at hourly resolution, from a test case study in the Northeastern United States. This is a geographical region characterized by substantial offshore wind resources and potential development plans. The data is explained below. 

\subsubsection{Weather Data}
Real-world measurements of 
weather conditions are extracted at the location of a potential offshore wind farm location in the New York Bight. The data is publicly available \citep{wave_data, aziz_ezzat_2025_15642047}, and contains hub-height, farm-level wind speed and wave height measurements collected by an ocean Lidar buoy deployed by the New York State Energy Research \& Development Authority (NYSERDA) to inform wind resource assessment and planning studies in this region. 

\subsubsection{Power System Data}
The electricity price and curtailment information are obtained by simulating the PJM Interconnection grid using the real-world 7-zone PJM dataset from \citet{icarus}.
The dataset reflects the current PJM generation portfolio and inter-zonal transmission limits. 
We augmented the dataset by adding three offshore wind farms corresponding to potential future projects in the region \citep{NJDEP_OffshoreWindProjects_2025}.
The simulation replicates PJM operations in a two-stage process.
For each day, we solve a first-stage (day-ahead) unit commitment (UC) model, using wind speed forecasts to determine wind power generation. 
A second-stage real-time (RT) model then receives the actual wind power generation, which may differ from the day-ahead forecast, and computes an updated cost-minimal generator dispatch based on the resources that are available from the UC decision. 
See \citet[Eqs. (1) and (2)]{ghazanfariharandi2025value} for a detailed description of the UC and RT models.

We use prices from two PJM zones
motivated by their different price profiles for comparison.
The more continental Pricing Zone~$1$ has predominantly positive prices over the simulation period, with an average electricity price of \si{\$}\SI{96.76}{\per\MWh}. 
Pricing Zone $2$ reflects the actual effective selling price at the offshore wind farm which is exposed to negative prices (about $10$\% of the time) 
and has a relatively lower average price of \si{\$}\SI{41.94}{\per\MWh}.

\subsubsection{Turbine Health Data}
Turbine-specific RULs are simulated by an exponential model \citep{gebraeel2008prognostic}: 
\begin{equation}\label{eq:exponential_degradation_model}
A_i(t) = \alpha_i \, e^{\beta_i t},
\end{equation}
where $A_i(t)$ denotes the degradation signal of turbine $i$ at time $t$. 
$\alpha_i$ and $\beta_i$ are the turbine-specific scaling parameter and degradation rate parameter, respectively. 
In practice, wind turbines either fail by degradation or shock. 
Statistics from the offshore wind turbine reliability literature suggest that degradation-driven failures account for roughly $70\%$ of failures, whereas unpredictable failures account for about $25\%$ \citep{Faulstich}. Thus, we also generate failure shocks according to a Poisson process such that $N_i(d) \sim \text{Poisson}\left(\gamma^s_i\cdot \frac{d}{365}\right)$, where $N_i(d)$
denotes the cumulative number of failure shocks at day $d$ for turbine $i$, and $\gamma^s_i$ is the yearly failure rate (by shock). 

\subsubsection{Power curve model}
The normalized power levels, $f_{t,i}$ and $f_{d,i}$, are obtained by mapping hub-height wind speed to power generation using a data-driven power curve model that has been previously validated on offshore wind data \citep{wind_power_transform}. Other data-driven power curve models that incorporate 
multiple weather parameters, in addition to wind speed, can be equivalently used if needed \citep{Ding}. 


\subsection{Experimental Setup}
We conduct two sets of case studies. Case Study I considers a toy example of a five-turbine offshore wind farm using electricity prices from Pricing Zone $1$. This simplified test case provides insight into the performance of GIMS.
Case Study II is a realistic setting of a $50$-turbine offshore wind farm. We run Case Study II twice using electricity price from Pricing Zones $1$ and $2$.
Table \ref{tab:pars} shows the parameter values in all of our experiments, which are largely based on prior studies in offshore wind farm maintenance \citep{STOCHOS, POSYDON}. We assume identical repair times across all turbines, with $\tau_i = \tau$ for all $i$. 
For asset health, we calibrate the exponential model parameters of \eqref{eq:exponential_degradation_model} to obtain physically plausible degradation paths that yield identical annual failure rates, $\gamma^d = \gamma_i^d$ and $\gamma^s = \gamma_i^s$ for all $i$, with $\gamma^d$ and $\gamma^s$ tuned to match reported failure statistics in the wind industry \citep{Dinwoodie, Carroll,elsahhar2025analysis}.  

\begin{table}[t]
    \caption{Parameter settings for the two sets of case studies}
    \label{tab:pars}
    \centering
    \renewcommand{\arraystretch}{0.8}
    \setlength{\tabcolsep}{4pt}
    \begin{adjustbox}{max width=\linewidth}
    \begin{tabular}{|c|c||c|c||c|c||c|c|}
        \hline
        Parameter & Value & Parameter & Value
        & Parameter & Value & Parameter & Value \\
        \hline
        $R$ & $12$ MW
        & $C^o$ & \$125/hour
        & $N^c$ & $2$ crews
        & $\gamma_d$ & $3.395$ \\

        $C^{\text{PM}}$ & \$15,000
        & $C^v$ & \$8,000/day
        & $N^w$ & $8$ hours
        & $\gamma_s$ & $1.213$ \\

        $C^{\text{CM}}$ & \$30,000
        & $T$ & $3$ hours
        & $N^o$ & $4$ hours
        & $\alpha$ & $0.05$ \\

        $C^\lambda$ & \$30/day
        & $\tau$ & $12$ hours
        & $\nu_{\max}$ & $15$ m/s
        & $\beta$ & $0.10$ \\

        $C^c$ & \$250/hour
        & $N^\lambda$ & $30$ days
        & $\eta_{\max}$ & $1.8$ m
        & & \\
        \hline
    \end{tabular}
    \end{adjustbox}
\end{table}

We implement GIMS in Python and solve it using Gurobi (v9.1.2) 
on a computer with an AMD Ryzen $9$ $7845$HX CPU 
and 16 GB of RAM.
A $15$-day rolling horizon is employed 
that advances daily over a 284-day evaluation period.
The average solution times for Case Studies I and II are $0.6$ and $5.6$ minutes, respectively, to a relative optimality gap of $\le0.01\%$. 

The following benchmark methods are used to evaluate the performance of GIMS: 
(\textit{i}) \textit{Corrective maintenance strategy (CMS)}, where maintenance is initiated in response to a failure, and no PM actions are considered. CMS is a na{\"i}ve benchmark used routinely in the maintenance scheduling literature and is considered as a lower bound baseline; 
(\textit{ii}) \textit{Corrective and curtailment maintenance strategy (CCMS)}: CCMS extends CMS using a cost-benefit analysis similar to the one proposed by ~\citet{Orla}. Preventive maintenance is allowed only when the cost-benefit analysis shows that the curtailment compensation provided by the grid operator is higher than the projected repair cost; (\textit{iii}) \textit{Holistic Opportunistic Strategy (HOST)}: This is an adaptation of a recently proposed maintenance optimization model by \citet{HOST}. HOST leverages the endogenous maintenance triggers (namely, asset health, weather, and maintenance resources) but, unlike GIMS, does not leverage exogenous, grid-informed maintenance triggers. 

We evaluate all models in terms of economic, operational, and logistical performance. Economic evaluation metrics include maintenance cost (\si{\$}), revenue (\si{\$}), and cost per unit of energy (\si{\$\per\MWh}). 
Operational evaluation metrics include time- and energy-based availability. Specifically, time-based availability is measured as 
the fraction of time during which turbines are operational. Importantly, curtailment-induced shutdowns are excluded.
In contrast, {energy-based availability} is measured as the fraction of achievable energy that is generated,  which is defined as the ratio between actual power output and the post-curtailed power limit.
Finally, logistical metrics include vessel utilization which represents the ratio of dispatched to chartered vessels, the number of vessel rentals, and the number of maintenance tasks per vessel which serves as another proxy for maintenance resource utilization. 

\section{Results and Analyses} \label{S:res_ana}
\subsection{Case Study I: The 5-Turbine Test Case}

As shown in Table~\ref{tab:benchmark}, GIMS clearly improves all evaluation metrics relative to CMS and CCMS. 
Meanwhile, the comparison with HOST reveals the gains realized by GIMS from tapping into exogenous grid-level information. GIMS reduces the total maintenance cost 
by $16.3\%$ relative to HOST. A comparable improvement is observed in the maintenance cost per unit of energy, which declines from \si{\$}\SI{2.42}{\per\MWh} to \si{\$}\SI{2.04}{\per\MWh} ($15.7\%$ reduction). The quantitative results indicate that grid signal visibility yields additional economic and operational benefits beyond the sole reliance on farm-level information. 
\begin{table}[]
\centering
\caption{Case Study I results (5 wind turbines). Bold-faced values denote best performance.}
\label{tab:benchmark}

\renewcommand{\arraystretch}{0.8}
\setlength{\tabcolsep}{4pt}

\begin{adjustbox}{max width=\linewidth}
\begin{tabular}{l|cccc}
\hline
& \multicolumn{4}{c}{\textbf{O\&M Strategy}} \\
\textbf{Performance Metric}
& CMS & CCMS & HOST & GIMS \\
\hline
Maintenance Cost (\si{\$}\si{M})
& 0.51 & 0.47 & 0.43 & \textbf{0.36} \\

Cost/Energy (\si{\$}\si{\per\MWh})
& 2.91 & 2.69 & 2.42 & \textbf{2.04} \\

Revenue (\si{\$}\si{M})
& 15.36 & 15.37 & \textbf{15.46} & \textbf{15.46} \\

Time-based Availability
& 99.4\% & 99.4\% & 99.4\% & \textbf{99.5\%} \\

Energy-based Availability
& 98.9\% & 98.9\% & 99.3\% & \textbf{99.4\%} \\

Vessel Rentals
& 16.0 & 15.0 & 13.0 & \textbf{9.0} \\

Vessel Utilization
& 68.8\% & 67\% & \textbf{100\%} & \textbf{100\%} \\

Maintenance/Vessel Ratio
& 0.69 & 0.73 & 1.23 & \textbf{1.56} \\
\hline
\end{tabular}
\end{adjustbox}
\end{table}

To support the quantitative results, we examine differences in how GIMS responds to market signals compared to HOST for a representative 5-day period in April 2020, as illustrated in Figure \ref{fig:5WT_maintenance vs revenue}. 
The black dashed curve in Figure~\ref{fig:5WT_maintenance vs revenue} represents the (hypothetical) hourly pre-curtailed revenue. 
The green shaded area denotes the (actual) post-curtailed revenue, corresponding to the effective power generation after accounting for curtailment. 
The gap between the dashed and shaded curves therefore captures the revenue reduction caused by curtailment. Maintenance actions for two wind turbines ($1$ and $2$) are then shown (corresponding to 40\% reduction in nominal power output of the wind farm). GIMS 
opportunistically schedules these two preventive maintenance tasks on April $15$, where the difference between the pre- and post-curtailed revenues is highest. This leverages the expected $54$\% curtailment of that day so that the 
maintenance outages (40\%) can be fully absorbed by curtailment and hence do not drive additional revenue losses. In contrast, 
HOST misses this opportunity
and instead defers these same two maintenance actions to April 17 (two days later), 
where the curtailment ratio drops to $23\%$, which is insufficient to offset the same $40\%$ reduction caused by two maintenance tasks. As a result, the uncovered $17.0\%$ reduction directly translates into additional revenue losses. 
\begin{figure*}
    \centering
    \includegraphics[angle=0, width=1\linewidth]{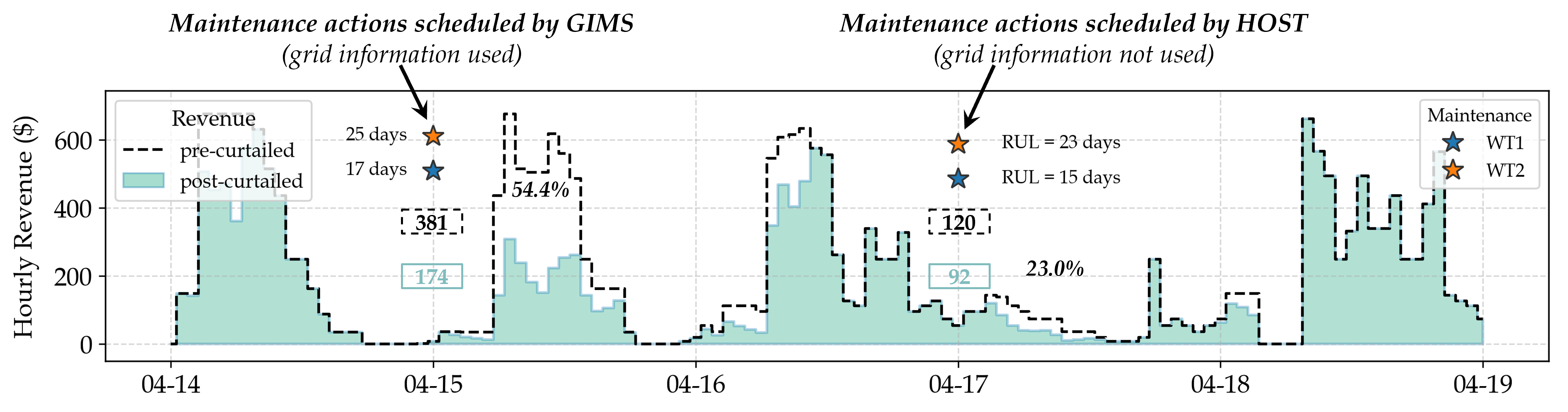}
    \caption{
    GIMS and HOST maintenance schedules in April 2020. Stars mark maintenance starts and are labeled with turbine RUL. At the maintenance date, the boxed values report the average hourly revenue, under pre-curtailed (dashed black) and post-curtailed (green) conditions. The percentages indicate the average curtailment ratio.}
    \label{fig:5WT_maintenance vs revenue}
\end{figure*}

\subsection{Case Study II: The 50-Turbine Test Case}
Table~\ref{tab:50WT_results_combined} reports the evaluation results of all methods under both pricing zones.

\subsubsection{Pricing Zone 1}
Similar to Case Study I, GIMS clearly outperforms CMS and CCMS across all metrics. Compared to HOST, GIMS achieves a substantial reduction in maintenance-related expenditures. Maintenance costs decrease from \si{\$}\SI{5.19}{M} to \si{\$}\SI{4.66}{M} ($10.2\%$ reduction). A comparable saving is observed in the cost per unit of energy, which declines from \si{\$}\SI{2.95}{\per\MWh} to \si{\$}\SI{2.65}{\per\MWh}. 
GIMS also requires the fewest vessel rentals, while attaining the highest vessel utilization rate among all strategies. The maintenance-per-vessel ratio is $2.42$, i.e., $19.8\%$ better than the best alternative, which highlights the effectiveness of GIMS in exploiting maintenance grouping effects to improve logistical efficiency.


\begin{table*}[t]
\centering
\caption{Performance comparison for Case Study II (50 wind turbines) under Pricing Zones 1 (left) and 2 (right). Bold-faced values denote best performance.}
\label{tab:50WT_results_combined}

\renewcommand{\arraystretch}{0.8}
\setlength{\tabcolsep}{4pt}

\begin{adjustbox}{max width=\textwidth}
\begin{tabular}{l|cccc||cccc}
\hline
\multirow[c]{2}{*}
& \multicolumn{4}{c||}{\textbf{Pricing Zone 1}}
& \multicolumn{4}{c}{\textbf{Pricing Zone 2}} \\
\textbf{Performance Metric}
& CMS & CCMS & HOST & GIMS
& CMS & CCMS & HOST & GIMS \\
\hline

Maintenance Cost (\si{\$}\si{M})
& 6.38 & 5.88 & 5.19 & \textbf{4.66}
& 6.28 & 5.70 & 4.83 & \textbf{4.44} \\

Cost/Energy (\si{\$}\si{\per\MWh})
& 3.65 & 3.36 & 2.95 & \textbf{2.65}
& 4.24 & 3.84 & 3.24 & \textbf{2.98} \\

Revenue (\si{\$}\si{M})
& 153.4 & 153.5 & \textbf{153.9} & \textbf{153.9}
& 57.6 & 57.7 & \textbf{58.0} & \textbf{58.0} \\

Time-based Availability
& 99.2\% & 99.2\% & 99.3\% & \textbf{99.4\%}
& 99.2\% & 99.2\% & \textbf{99.4\%} & \textbf{99.4\%} \\

Energy-based Availability
& 98.7\% & 98.7\% & \textbf{99.1\%} & \textbf{99.1\%}
& 83.5\% & 83.7\% & \textbf{84.1\%} & \textbf{84.1\%} \\

Vessel Rentals
& 150 & 146 & 95 & \textbf{72}
& 149 & 146 & 84 & \textbf{70} \\

Vessel Utilization
& 68.7\% & 63.7\% & 97.9\% & \textbf{100\%}
& 63.8\% & 62.3\% & 98.8\% & \textbf{100\%} \\

Maintenance/Vessel Ratio
& 0.96 & 1.08 & 2.02 & \textbf{2.42}
& 0.89 & 0.96 & 2.12 & \textbf{2.33} \\
\hline
\end{tabular}
\end{adjustbox}
\end{table*}

\subsubsection{Pricing Zone 2}
The main conclusions remain consistent with those observed in the Pricing Zone $1$, but with additional insights from the more frequent presence of negative prices. 
Specifically, it can be observed that maintenance costs under Pricing Zone $2$ are lower than those observed under Zone $1$ for both HOST and GIMS. This is attributed to the presence of negative electricity prices, during which the model effectively suspends power output until prices rebound. 
As a result, turbines 
require fewer maintenance actions. 
Yet, temporarily suspending generation during negative price periods also reduces total energy output. Consequently, the maintenance cost per unit energy increases in Pricing Zone 2 relative to Pricing Zone $1$. 


\section{Conclusions and Future Directions} \label{S:concl}

Grid-triggered maintenance signals create under-explored opportunities for improving offshore wind farm economics. This paper proposes GIMS, a grid-informed maintenance optimization framework that explicitly incorporates generation curtailment and electricity pricing signals into maintenance scheduling in offshore wind farms. 
By incentivizing the scheduling of maintenance actions during periods of high generation curtailment or economically unattractive periods, GIMS has the potential to reduce overall production and revenue losses significantly. Realistic case studies on offshore wind farms with $5$ and $50$ turbines suggest that GIMS consistently outperforms conventional and sophisticated maintenance models in terms of operational economics, asset availability, and resource utilization, highlighting the potential influence of grid-level information on the operational management of offshore wind farms. 

A key point to highlight is that the performance of GIMS will critically depend on high-quality market signal forecasts. Inaccurate predictions of curtailment or pricing signals may lead to suboptimal scheduling decisions and, in extreme cases, worse performance relative to simpler strategies. In this work, we have assumed, perhaps unrealistically so, that these forecasts are available and are sufficiently accurate, since the primary goal was to demonstrate the potential value of integrating such information into O\&M decision-making. This, nevertheless, points to several promising directions for future research. 

\textit{First}, accurate short- and medium-term forecast models for curtailment and pricing signals are needed. While the literature on the latter is reasonably rich and active (See, e.g., \cite{weron2014electricity, aziz2026iise}), we find that there curtailment forecasting is largely under-studied, with the exception of some very recent efforts \citep{memmel2023curtailment, Galsim2026CurtailmentForecasting}. This is perhaps due to the historically low levels of generation curtailment in most electricity systems world-wide. Indeed, curtailment is a problem of the future \citep{UScurtail, UKcurtail2024}. As curtailment levels continue to rise worldwide as driven by large grid decarbonization initiatives and evolving grid conditions, we anticipate prediction models thereof to mature, making data-driven decision models that make use of such information, as in this paper, practically relevant. 

\textit{Second,} and related to the point above, integrating these market signal forecasts into a stochastic optimization framework (see, e.g., \cite{POSYDON}) would allow maintenance decisions to account for various sources of farm-induced uncertainties (e.g., asset health forecasts, weather predictions) and equally importantly, grid-induced uncertainties from market and system signals. The latter is especially important to unlock meaningful gains from an opportunistic grid-informed maintenance model, since curtailment and pricing signals are formed by the interaction of complex drivers and may potentially carry over substantial uncertainties that need to be rigorously accounted for. 
\textit{Third,} the proposed approach takes the view of a wind power producer who is responsive to the electric grid. This constitutes a step toward coordinated maintenance frameworks in which grid- and farm-level objectives can be jointly co-optimized, perhaps through a bi-level optimization formulation, enabling a rigorous, two-way coupling between farm and system operations.

\if0\blind{
\section*{Data availability statement}
Details to reproduce the results of this work are included in the reproducibility report. The codes and data will be made publicly available upon the paper acceptance. 
\fi

\if0\blind{
\section*{Acknowledgments}
This work is supported in part by the National Science Foundation (ECCS, Award \# 2114422).} \fi

\section*{Acknowledgments}
This work is supported in part by the U.S. National Science Foundation (ECCS, \# 2114422).

\section*{Declaration of Generative AI Use}
While writing this manuscript, the authors used ChatGPT (OpenAI, GPT-5.6 Luna) to enhance the clarity, language, and grammar of certain statements. They subsequently reviewed and revised all material and take full responsibility for the published work.

\section*{Disclosure Statement}
The authors report there are no competing interests to declare.

\appendix

\section{Appendix: Proof of Proposition~\ref{prop:aware_vs_blind}}
\label{app:proof_aware_blind}

\begin{proof}
Define \(h(v)=\mathbb E[\Delta V\mid V=v]\), and let
\(V_{d^{\mathrm I}}=V_{(1)}=\min_{d\in\mathcal D}V_d\).
Because the daily pairs \((C_d,K_d)\) are independent across
candidate days,
\(\mathbb E[\Delta V^{\mathrm I}\mid V_1,\ldots,V_{N_{\mathcal D}}]
=h(V_{(1)})\).
Therefore,
\(\mathbb E[\Delta V^{\mathrm I}]=\mathbb E[h(V_{(1)})]\).
Under the grid-blind strategy, the selected day is distributed as
a generic candidate day, so
\(\mathbb E[\Delta V^{\mathrm B}]
=\mathbb E[\Delta V]=\mathbb E[h(V)]\).
Since \(V_{(1)}\le V_1\) for every realization,
Proposition~\ref{prop:aware_vs_blind} follows if \(h(v)\) is
nondecreasing. It remains to establish this monotonicity.

We can write \(\Delta V=Vq(C)\), where
\(q(c)=\frac{\Delta G/P^{\mathrm{nom}}-c}{1-c}
\mathbb{I}\{c<\Delta G/P^{\mathrm{nom}}\}\)
for \(0\le c<1\), and \(q(1)=0\).
For \(c<\Delta G/P^{\mathrm{nom}}\),
\(q'(c)=(\Delta G/P^{\mathrm{nom}}-1)/(1-c)^2\le 0\),
and \(q(c)=0\) for \(c\ge\Delta G/P^{\mathrm{nom}}\).
Thus, \(q(c)\) is nonincreasing in \(c\).
Conditioning \(\Delta V=Vq(C)\) on \(V=v\) gives
\(h(v)=v\,\mathbb E[q(C)\mid V=v]\).

We now show that \(\mathbb E[q(C)\mid V=v]\) is nondecreasing
in \(v\). By Bayes' rule,
\(f_{C\mid V}(c\mid v)
=f_{V\mid C}(v\mid c)f_C(c)/f_V(v)\).
Consider any \(v_2>v_1\). Under the condition stated in
Proposition~\ref{prop:aware_vs_blind}, for any \(c_2>c_1\),
\(\frac{f_{V\mid C}(v_2\mid c_2)}
        {f_{V\mid C}(v_2\mid c_1)}
\le
\frac{f_{V\mid C}(v_1\mid c_2)}
     {f_{V\mid C}(v_1\mid c_1)}\).
Equivalently,
\(\frac{f_{V\mid C}(v_2\mid c_2)}
        {f_{V\mid C}(v_1\mid c_2)}
\le
\frac{f_{V\mid C}(v_2\mid c_1)}
     {f_{V\mid C}(v_1\mid c_1)}\).
By Bayes' rule, it follows that
\(f_{C\mid V}(c\mid v_2)/f_{C\mid V}(c\mid v_1)\)
is nonincreasing in \(c\). Since both conditional densities
integrate to one,
\(\Pr(C\le c\mid V=v_2)\ge\Pr(C\le c\mid V=v_1)\)
for all \(c\in[0,1]\), so curtailment is conditionally shifted
toward lower values at higher revenue. Because \(q(c)\) is
nonincreasing, this implies
\(\mathbb E[q(C)\mid V=v_2]
\ge \mathbb E[q(C)\mid V=v_1]\).

Since \(q\ge 0\) and revenue is nonnegative,
\(h(v)=v\,\mathbb E[q(C)\mid V=v]\) is therefore nondecreasing.
Finally, \(V_{(1)}\le V_1\) implies \(h(V_{(1)})\le h(V_1)\).
Taking expectations and using the expected-loss identities above gives
\(\mathbb E[\Delta V^{\mathrm I}]
=\mathbb E[h(V_{(1)})]
\le \mathbb E[h(V_1)]
=\mathbb E[\Delta V^{\mathrm B}]\).
\end{proof}

\bibliographystyle{chicago-disambig}
\spacingset{1}
\bibliography{IISE-Trans}
	
\end{document}